\documentclass[11pt]{article}
\usepackage[margin=1in]{geometry}
\usepackage{amsmath,amssymb,amsthm,mathtools}
\usepackage{hyperref}
\usepackage{enumitem}
\usepackage{algorithm,caption}
\usepackage[many]{tcolorbox}
\usepackage{nicefrac}
\usepackage{authblk}
\usepackage{microtype}

\newcommand{\Tr}{\operatorname{Tr}}

\newcommand{\E}{\mathbb{E}}
\newcommand{\OPT}{\mathrm{OPT}}
\newcommand{\SOL}{\mathrm{SOL}}
\newcommand{\SDP}{\mathrm{SDP}}
\newcommand{\Id}{\mathbb{I}}
\newcommand{\sgn}{\operatorname{sign}}
\newcommand{\ket}[1]{\left|#1\right\rangle}

\newcommand{\ip}[2]{\left\langle #1,#2\right\rangle}

\newcommand{\htfimp}{H_{\mathrm{TFIM}'}}

\theoremstyle{plain}
\newtheorem{theorem}{Theorem}
\newtheorem{lemma}{Lemma}

\newtheorem{corollary}{Corollary}
\theoremstyle{definition}
\newtheorem{definition}{Definition}
\newtheorem{remark}{Remark}

\title{\textbf{Product-State Approximation Algorithms for the Transverse Field Ising Model}}

\author{Vincenzo Lipardi\thanks{vincenzo.lipardi@maastrichtuniversity.nl}}
\author{David Mestel\thanks{david.mestel@maastrichtuniversity.nl}}
\author{Georgios Stamoulis\thanks{georgios.stamoulis@maastrichtuniversity.nl}}
\affil{Department of Advanced Computing Sciences, Maastricht University}

\date{}
\begin{document}
\date{}
%\address{Department of Advanced Computing Sciences \\ Maastricht University, The Netherlands}
\maketitle

\begin{abstract}
We study classical polynomial-time approximation algorithms for the transverse-field Ising model (TFIM) Hamiltonian, allowing a mixture of ferromagnetic and anti-ferromagnetic interactions between pairs of qbits, alongside transverse field terms with arbitrary non-negative weights.

Our main results are a series of approximation algorithms (all approximation ratios with respect to the true quantum optimum):
(i) a simple maximum of two product state rounding algorithm achieving an approximation ratio $\gamma\approx 0.71$ ,
(ii) a strengthened rounding, inspired by the anticommutation property of the two $X_i, Z_iZ_j$ observables achieving ratio $\gamma\approx 0.7860$, and (iii) a further improvement by interpolation achieving ratio $\gamma \approx 0.8156$.
%these can be compared to the UGC-hardness bound of $\alpha_{\mathrm{GW}}\approx 0.8786$ inhereted from classical {\sc MaxCut}. 
We also give an explicit (purely ferromagnetic) TFIM instance on three qbits for which every product state achieves at most $169/180\approx 0.9389$ of the true optimum, yielding an upper bound for all algorithms producing product state approximations, even in the purely ferromagnetic case.
\end{abstract}

%\tableofcontents

% ================================================================
\section{Introduction}
% ==========================================================

\textit{Local Hamiltonians} are operators associated with the total energy of quantum many-body systems and are represented by Hermitian matrices. Even though general Hamiltonians can have exponentially large descriptions with respect to the number of interacting particles, local Hamiltonians can be succinctly represented as the summation of only the non-trivial locally interacting terms, which are usually very small (2 or 3). 

The \textit{Local Hamiltonian problem} asks to determine whether the \textit{ground state energy} (associated with the smallest eigenvalue of the operator) of a given local Hamiltonian is below a certain threshold $\alpha$ or above a different threshold $\beta$, with a promise that one of these two conditions holds. This problem is the quantum analogue of classical constraint satisfaction problems \cite{DBLP:journals/siamcomp/CubittM16}. In a seminal work \cite{kitaev2002classicalQuantumComputation} Kitaev showed that computing the ground state energy of such Hamiltonians is a \textbf{QMA}-hard problem \cite{bookatz2014QMAcomplete}, even when the interactions are restricted in an one-dimensional line and the dimension of each particle is bounded by 12 \cite{kempe2006complexityLocalHamiltonian, Aharonov_2009, aharonov2002quantumnpsurvey}. Such problems are believed to be intractable even for quantum computers (if and when they might be built).  

In recent years, the study of Local Hamiltonians has also attracted growing interest within the computer science community.  Given the fundamental importance and computational intractability of the problem, research efforts have focused on finding \textit{approximation algorithms}: efficient algorithms designed to estimate the ground-state energy by producing low-energy states that approximate the true ground state within provable performance guarantees.  \cite{BranDaoHarrow2016, BravyiGossetKoenigTemme2019, GharibianKempe2012, GharibianParekh2019, HallgrenLeeParekh2020, ParekhThompson2021, ParekhThompson2022, AnshuGossetMorenz2020, King2023, ParekhThompsonICALP2021, TakahashiEtAl2023, WattsEtAl2023, lee_et_al:LIPIcs.ICALP.2024.105, carlson2024approximationalgorithmsquantummaxdcut, gribling2025improvedapproximationratiosquantum, DBLP:journals/qic/BansalBT09, jorquera2025monogamyentanglementboundsimproved, kannan2024quantumapproximateoptimizationalgorithm, apte202508395approximationalgorithmeprproblem} These algorithms typically output product states or states that can be efficiently represented on a classical device, such as Matrix Product States. Although the `ansatzes' from which these states are drawn are not fully general, they are frequently relevant in practice, for example in the efficient simulation of material properties in quantum chemistry.

\medskip
The Transverse Field Ising model (TFIM) is a well-known and well-studied quantum spin system consisting of competing  \textit{Ising} interactions and  \textit{transverse fields}. An instance consisting of $n$ interacting particles (qubits) can be represented by a simple, non-directed simple graph $G=(V,E)$ with vertex set $V$, one for each particle, and edge set $E$, one for each pair of interacting particles. The TFIM Hamiltonian is defined as

\begin{equation}\label{eq:tfim}
H_{\mathrm{TFIM}} = \sum_{\{i,j\}\in E}  w_{ij} J_{ij} Z_i Z_j - \sum_{i\in V}  h_i X_i,
\qquad  J_{ij}\in \{-1,+1\}, w_{i,j}, h_i \geq 0.
\end{equation}
\noindent
where $X_i,Z_i$ are Pauli operators on qubit $i$, $J_{ij}$ are arbitrarily signed couplings and $h_i$ are the single-qubit transverse-field strengths.

This paper studies \emph{multiplicative} approximation algorithms, which are subtle in quantum Hamiltonians because, for instance, the ground state energy of the Hamiltonian as written in (\ref{eq:tfim}) could be negative or even zero.  More generally, adding a constant multiple of the identity changes multiplicative approximability even though it does not change the ground state. This is the same classical phenomenon behind ``Ising energy'' versus ``Max-Cut'' objectives, and is explicitly discussed in the quantum setting (i.e., Gharibian-Parekh \cite{GharibianParekh2019} for the Heisenberg model). Consequently, we fix an equivalent maximization version of the problem in which the Hamiltonian is a weighted sum of local \emph{projections} (taking values between 0 and 1), and we approximate its maximum eigenvalue.

Given the equivalent formulation of the TFIM problem, we design classical polynomial-time algorithms outputting explicit \emph{product states} (unentangled states) achieving a guaranteed fraction of the true quantum optimum. Product states are a natural restricted ansatz (mean-field) and are algorithmically attractive due to efficient description and evaluation. Indeed, a product state can be described with polynomially many bits (in the input instance size), as opposed to entangled states which typically require exponentially many classical bits to be described. A central question is:  \emph{how well can product states approximate the true optimum of TFIM instances?}

\paragraph{Our Contributions.}
We give three constant-factor approximation algorithms and one product state upper bound:
\begin{enumerate}
\item One simple, baseline algorithm achieving $0.7154$ that takes the maximum of two solutions: one that completely ignores the Ising terms and optimizes the field terms, and the complementary that ignores the field terms and optimizes the Ising terms. This is easy to analyze because of the competing nature of the two constraints (Ising $ZZ$ terms vs field $X$ terms).
\item 
We introduce an SDP relaxation strengthened by \textit{second order conic} (SOC) inequalities capturing the local incompatibility (anticommutation) between the $X_i$ and $Z_iZ_j$ observables.
We again define two rounded product-state candidates and return the maximum of the two: a standard hyperplane rounding, and a second rounding that uses the SOC tradeoff to lower bound the $ZZ$-correlation that may remain given the SDP's $x$-mass. This yields an improved approximation ratio of $\approx 0.7860$.
\item The third algorithm performs interpolation on the $x$ variables obtained by the SOC-SDP: instead of fully trusting the SDP solution for the $x$ variables, as in the previous solution, we scale them with a factor $q$. This might ``tilt" the solution more towards the $ZZ$ terms. By again carefully bounding each term's contributions, and optimizing over the parameter $q$, we obtained an improved $0.8156$-approximation algorithm.
\item Finally, we construct a simple example with an explicit product-state upper bound of ($169/180$) i.e., demonstrating that no algorithm that outputs products states can achieve better than $169/180 \approx 0.9388$-approximation ratio against the true optimum.
\end{enumerate}

All guarantees hold for \emph{arbitrary sign patterns} $J_{ij}\in\{+/- 1\}$ (including frustrated spin glasses) and weighted constraints. Also these approximation guarantees can be compared to the UGC-hardness bound of $\alpha_{\mathrm{GW}}\approx 0.8786$ inherited from classical {\sc MaxCut}. 

\paragraph{Scope relative to prior work.}
Bravyi and Hastings \cite{bravyi2014complexityquantumisingmodel} showed that estimating the ground state energy of TFIM on degree-3 graphs is a complete problem for the complexity class \textbf{Stoq-QMA}, a restriction of the class \textbf{QMA} to stoquastic Hamiltonians. Bravyi and Gosset  \cite{PhysRevLett.119.100503} give randomized approximation schemes for the \emph{partition function} (hence additive free-energy/ground-energy estimates) for certain ferromagnetic stoquastic models including ferromagnetic TFIM after a basis change. An instance is ferromagnetic when the couplings $J_{ij}$ are non-negative. Our approach, on the other hand, works for arbitrary sign patterns (and indeed for arbitrary weights). Such additive schemes do not address worst-case  \emph{multiplicative} approximation of the normalized Hamiltonian, nor they seem to extent to arbitrary signed/frustrated couplings. In contrast, our results not only apply to signed instances but also provide explicitly constructed product states with guaranteed multiplicative value in the constraint objective. Works on quantum Heisenberg/Quantum MaxCut Hamiltonians study different local terms. TFIM has commuting $ZZ$ edges but noncommuting $X$ fields, requiring a relaxation that captures the local $X$ vs $ZZ$ tradeoff.

% ==========================================================
\section{The Model and its relaxation}
% ==========================================================

On one qbit, let
\[
X=\begin{pmatrix}0&1\\1&0\end{pmatrix},
\qquad
Z=\begin{pmatrix}1&0\\0&-1\end{pmatrix},
\qquad
\Id=\begin{pmatrix}1&0\\0&1\end{pmatrix}.
\]
On $n$ qubits, $X_i$ denotes $X$ acting on qbit $i$ tensored with identity on others; similarly for $Z_i$.
Pauli operators satisfy $X_i^2=Z_i^2=\Id$, and they anticommute when acting on the same qubit $X_i Z_i = - Z_i X_i$, while they commute when acting on different qbits.

A (mixed) quantum state on $n$ qbits is a density matrix $\rho\succeq 0$ with $\Tr(\rho)=1$.
For a Hermitian observable $A$, its expectation value is denoted $\langle A\rangle_\rho := \Tr(\rho A)$.
A \emph{product state} is $\rho=\rho_1\otimes\cdots\otimes\rho_n$. Any single-qubit state can be written in the Bloch form
\[
\rho_i = \frac12\left(\Id + x_i X + y_i Y + z_i Z\right),
\quad \mbox{ such that } x_i^2+y_i^2+z_i^2\le 1.
\]
We note that $\langle X_i\rangle_{\rho_i}=x_i$ and $\langle Z_i\rangle_{\rho_i}=z_i$.

\subsection{Formulating the approximation problem }

A fundamental computational task is to approximate the ground-state energy $\lambda_{\min}(H_{\mathrm{TFIM}})$ and to output a quantum state (in our case, a \textit{product state}) whose energy is close to optimum (ground state). For classical approximation algorithms, however, the ``usual'' multiplicative approximation ratio for minimization is not well suited for quantum Hamiltonians because the ground energy can be negative, or even zero. To obtain a robust notion of approximation we transform our $H_{\mathrm{TFIM}}$ to an equivalent \emph{maximization} formulation built from positive semidefinite (PSD) ``constraint'' operators. This affine shift does not change the physics of the model, it only transforms the Hamiltonian to one where the approximation mechanics make sense. The idea is to replace each local term by a PSD operator that rewards the \emph{locally energy-minimizing} eigenspace, so that the global objective becomes a sum of nonnegative constraint values. This provides a natural scale for multiplicative approximation guarantees.

For an edge $\{i,j\}\in E$, the operator $J_{ij} Z_i Z_j$ has eigenvalues $+1/-1$. The contribution of this edge to \eqref{eq:tfim} is minimized when $J_{ij}Z_iZ_j = -1$. This can happen in two ways: either $J_{ij} = +1$ in which case the corresponding vertices would like to be \emph{anti}-aligned (an `antiferromagnetic' interaction), or $J_{ij} = -1$ in which case $i$ and $j$ would like to be aligned in the $z$-direction (a `ferromagnetic' interaction).  We define the PSD projector onto this minimizing eigenspace:
\begin{displaymath}
\Pi_{ij}^{Z} = \frac{1}{2}\bigl(\Id - J_{ij} Z_i Z_j\bigr).
\end{displaymath}
This operator has eigenvalues $\{0,1\}$ and equals $1$ exactly when the edge is ``satisfied'' in the energy-minimizing sense.

Similarly for the field ``constraints", for a vertex $i\in V$ the single-qbit operator $-h_i X_i$ is minimized when $X_i=+1$ (the $|+\rangle$ eigenstate of $X$). We define the PSD projector onto this minimizing eigenspace:
\begin{displaymath}
\Pi_{i}^{X} = \frac{1}{2}\bigl(\Id + X_i\bigr).
\end{displaymath}
Again $\Pi_{i}^{X}$ has eigenvalues $\{0,1\}$.

We can now define the equivalent tranfsformed Hamiltonian by summing these PSD constraints with the same nonnegative weights as in the TFIM:
\begin{equation}
\label{eq:hcsp}
H_{\mathrm{TFIM'}} = \sum_{\{i,j\}\in E} w_{ij} \Pi_{ij}^{Z} + \sum_{i\in V} h_i\, \Pi_{i}^{X}.
\end{equation}

By construction we have that $H_{\mathrm{TFIM'}}\succeq 0$ and that  $\langle\psi|H_{\mathrm{TFIM'}}|\psi\rangle$ is the energy of the state $|\psi\rangle$. By expanding \eqref{eq:hcsp} we immediately see that that maximizing $H_{\mathrm{TFIM'}}$ is exactly equivalent to minimizing $H_{\mathrm{TFIM}}$, up to an affine shift. Indeed,

\begin{eqnarray*}
H_{\mathrm{TFIM'}} & = & \sum_{\{i,j\}\in E} w_{ij}\Pi_{ij}^{Z} + \sum_{i\in V} h_i\, \Pi_{i}^{X}  \\ 
& = & \sum_{\{i,j\}\in E}\frac{w_{ij}}{2}\bigl(\Id - J_{ij} Z_i Z_j\bigr) + \sum_{i\in V} h_i\frac{1}{2}\bigl(\Id + X_i\bigr) \\
& = & \frac{1}{2} \Big( \sum_{\{i,j\} \in E} w_{ij} \Big)\Id -\frac{1}{2}\sum_{\{i,j)\}\in E} J_{ij} Z_i Z_j + \frac{1}{2}\Bigl(\sum_{i\in V} h_i\Bigr)\Id
+\frac{1}{2}\sum_{i\in V} h_i X_i \\ 
& = & \frac{1}{2}\Bigl(\sum_{\{i,j\} \in E} w_{ij}+\sum_{i\in V} h_i\Bigr)\Id
-\frac{1}{2}H_{\mathrm{TFIM}}.
\end{eqnarray*}

\noindent
Taking maximum eigenvalues gives 
\begin{align*}
\lambda_{\max} \bigl(H_{\mathrm{TFIM'}}\bigr) & = \frac{1}{2}\Bigl(\sum_{\{i,j\} \in E} w_{ij}+\sum_{i\in V} h_i\Bigr) -\frac{1}{2}\lambda_{\min} \bigl(H_{\mathrm{TFIM}}\bigr) \Leftrightarrow \\ 
\lambda_{\min}\bigl(H_{\mathrm{TFIM}}\bigr) & =  \Bigl(\sum_{\{i,j\} \in E} w_{ij}+\sum_{i\in V} h_i\Bigr) - 2\lambda_{\max} \big( H_{\mathrm{TFIM'}}\big).
\end{align*}
Moreover, $H_{\mathrm{TFIM}}$ and $H_{\mathrm{TFIM'}}$ have the same eigenvectors, so any state that approximately maximizes \eqref{eq:hcsp} immediately yields an approximate ground state for \eqref{eq:tfim}.

The shift from \eqref{eq:tfim} to \eqref{eq:hcsp} serves two purposes. First, it gives us an objective that is a nonnegative sum of PSD local terms so that we can talk about multiplicative approximation ratios in a meaningful way. Second, the constraint form \eqref{eq:hcsp} reveals the competing and incompatible local preferences of the transverse-field and the Ising (edge) terms, which would be the core of our Second Order Programming relaxations in subsequent sections.  In the remainder of the paper we design classical algorithms that, given a signed TFIM instance,  round a suitable semidefinite relaxation to produce product states with provable constant factor guarantees for $\lambda_{\max}(H_{\mathrm{TFIM'}})$, and hence for the TFIM ground-state energy.

\begin{remark} 
Because $Z_iZ_j$ has eigenvalues $\pm 1$, the operator $J_{ij}Z_iZ_j$ has eigenvalues $\pm J_{ij}$. Therefore $\Pi^Z_{ij}$ has eigenvalues $(1\pm J_{ij})/2$.
Thus $\Pi^Z_{ij}$ is a \emph{projector} (eigenvalues $\{0,1\}$) if and only if $J_{ij}\in\{\pm1\}$. If one only needs PSD constraints (not necessarily projectors), then it suffices that $|J_{ij}|\le 1$; however the cleanest tranformation uses $\{\pm1\}$ signs and separates weights/magnitudes into $w_{ij}$.
\end{remark}

\begin{remark}
Allowing arbitrary signs $J_{ij}\in\{\pm1\}$ includes frustrated ``spin-glass'' instances. A standard invariant is the sign product along a cycle: if for some cycle $C$,
$\prod_{(i,j)\in C} J_{ij}=-1$, then the instance is frustrated (no gauge transform can make all signs $+1$).
\end{remark}

\begin{remark}
In the case where all $h_i=0$ and all $J_{ij}=-1$ we have that $\lambda_{\max}(\htfimp) = \textsc{MaxCut}(G)$.  Since the latter is known to be UGC-hard to approximate with approximation ratio better than $\alpha_{\mathrm{GW}}\approx 0.8786$~\cite{khot2004inapprox}, we have the same inapproximability result for $\lambda_{\max}(\htfimp)$.
\end{remark}

%\begin{remark}
%Finally, we note that we may also accommodate Ising weights (not only signed couplings) in \ref{eq:hcsp}. Indeed, our SDP/SOC defined below works also for arbitrary Ising weights $w_{ij} \geq 0 $ i.e., it considers the weighted TFIM
%\[
%H_{\mathrm{TFIM'}}^w = \sum_{\{i,j\}} w_{ij} \Pi_{ij}^Z + \sum_{i\in V} h_i \Pi_i^X.
%\]
%\end{remark}

% ==========================================================
\section{Warm-up: a $0.7154$ two-candidate product state approximation algorithm}
We start with a simple algorithm that already achieves $0.7154$ approximation guarantee. In the next section we will refine it by explicitly accounting for the incompatibility between the two types of constraints (Ising vs. field), i.e., the fact that the corresponding observables anti-commute. The present algorithm simply defines two easy product state solutions: one that tries to gain as much as possible on the Ising terms while essentially ignoring the field terms, and a second one that does the opposite i.e.,  ignores the Ising terms but maximizes the field terms. We then output the better of the two. Recall that
\[
H_{\mathrm{TFIM'}} = \sum_{\{i,j\}\in E}w_{ij}\Pi^{Z}_{ij} + \sum_{i\in V} h_i \Pi_i^X,
\qquad \mathrm{OPT} =  \lambda_{\max}(H_{\mathrm{TFIM'}}).
\]
We also define
\[
H = \sum_{i\in V} h_i,
\qquad
H_Z = \sum_{\{i,j\}\in E}w_{ij }\Pi^{Z}_{ij},
\qquad
\mathrm{OPT}_Z = \lambda_{\max}(H_Z).
\]
Since $H_Z$ is diagonal in the computational ($Z$) basis, $\mathrm{OPT}_Z$ is achieved by a $Z$-basis product state (i.e., a classical $\{\pm 1\}$ assignment). Moreover, each $\Pi^Z_{ij}\preceq \Id$, so $\mathrm{OPT}_Z\le |E|$.

%\begin{description}
%\item[The Field Solution:] 
\paragraph{The Field Solution:} The field projector is $\Pi_i^X=\frac12(\Id +X_i)$, which is maximized when $X_i$ has eigenvalue $+1$, i.e., on the state $\ket{+}=\frac{\ket{0}+\ket{1}}{\sqrt{2}}$. We therefore define the field solution $\SOL^F$ to be the product state $\ket{+}^{\otimes n} = \ket{\psi}$. On this state, every field term is fully satisfied:
\[
\langle \psi | h_i \Pi_i^X | \psi \rangle = h_i .
\]
On the other hand, $\langle \psi | Z_i Z_j | \psi\rangle = 0$ and so
\[
\langle \psi | \Pi^Z_{ij} | \psi \rangle = \frac{1}{2} \qquad\text{for every }\{i,j\}\in E,
\]
independent of $J_{ij}$. Therefore,
\[
\SOL^F = \langle \psi | \htfimp | \psi  \rangle =  \frac{\sum_{\{i,j\} \in E} w_{ij}}{2} + H \geq \frac{\mathrm{OPT}_Z}{2} + H,
\]
where the last inequality uses $\mathrm{OPT}_Z\le |E|$.

%\item[The Ising Solution:] 
\paragraph{The Ising Solution:} The Ising solution $\SOL^I$ complements the above by focusing on the edge clauses. Maximizing $\sum_{\{i,j\}}\Pi^Z_{ij}$ over $ZZ$-basis product states is exactly the classical \emph{signed MaxCut} (equivalently Max2XOR) objective: if $J_{ij}=+1$ the clause prefers anti-alignment, and if $J_{ij}=-1$ it prefers alignment. The Goemans--Williamson hyperplane rounding applies unchanged in this setting (it is the same SDP and the same rounding guarantee), giving the approximation ratio $\alpha_{\mathrm{GW}}\approx 0.878567$. Thus, we can compute in polynomial time a $ZZ$-basis product state $\rho^I$ such that
\[
\Tr[\rho^I H_Z] \geq  \alpha_{\mathrm{GW}}\,\mathrm{OPT}_Z.
\]
For the field terms, since $\Tr[\rho^I X_i] = \langle X_i\rangle_{\rho^I}=0$ on any $Z$-basis state, we have
\[
\Tr[\rho^I\, h_i \Pi_i^X]= h_i \cdot \frac{1+\langle X_i\rangle_{\rho^I}}{2} = \frac{h_i}{2},
\]
and so we get
\[
\SOL^I = \Tr[\rho^I H_{\mathrm{TFIM'}}] = \Tr[\rho^I H_Z] + \sum_i \Tr[\rho^I h_i\Pi_i^X]
\geq  \alpha_{\mathrm{GW}}\,\mathrm{OPT}_Z + \frac{H}{2}.
\]
%\end{description}

We output the better of the two product states (i.e., the larger value between $\SOL^F$ and $\SOL^I$).

\begin{theorem}
Taking the maximum of the two product state solutions above gives an approximation algorithm for TFIM with approximation ratio $1-1/4\alpha_{\mathrm{GW}} \approx 0.7154$.
\end{theorem}

\begin{proof}
First, by subadditivity of the maximum eigenvalue,
\[
\mathrm{OPT}
=\lambda_{\max} \Big(H_Z+\sum_i h_i\Pi_i^X\Big) \leq \lambda_{\max}(H_Z)+\lambda_{\max} \Big(\sum_i h_i\Pi_i^X\Big) = \mathrm{OPT}_Z + H,
\]
since $\lambda_{\max}(\Pi_i^X)=1$ and the $\Pi_i^X$ act on different qbits.

Define \[t=\frac{\OPT_Z}{H+\OPT_Z} \leq \frac{\OPT_Z}{\OPT},\] and hence correspondingly $1-t=\frac{H}{H+\OPT_Z} \leq \frac{H}{\OPT}$. Then we have
\[\frac{\SOL^I}{\OPT} \geq \alpha_{\mathrm{GW}} \cdot t + \frac{1}{2}(1-t) = \frac{1}{2} + t\left(\alpha_{\mathrm{GW}}-\frac{1}{2}\right)\]
and \[\frac{\SOL^F}{\OPT}\geq \frac{t}{2} + 1-t = 1- \frac{t}{2}.\]
Hence the returned value satisfies 
\[
\frac{\max\{\SOL^I,\SOL^F\}}{\mathrm{OPT}} \geq \min_{t\in [0,1]} 
\max\left\{\frac{1}{2} + t\left(\alpha_{\mathrm{GW}}-\frac{1}{2}\right), 1- \frac{t}{2}\right\}.
\]
The first argument is an increasing function of $t$ and the second a decreasing function, so the minimum is at their intersection with $t=t^*=\frac{1}{2\alpha_{\mathrm{GW}}}$ and approximation ratio 
\[1-\frac{t^*}{2} = 1- \frac{1}{4\alpha_{\mathrm{GW}}} \approx 0.7154457.\]

% If $\mathrm{OPT}_Z=0$, then $\mathrm{OPT} = \SOL^G = H$  so the approximation ratio is $1$. We can therefore assume $\mathrm{OPT}_Z>0$ and define
% \[
% t =\frac{H}{\mathrm{OPT}_Z} \geq  0.
% \]
% Using the lower bounds proved above and $\mathrm{OPT}\le \mathrm{OPT}_Z+H=(1+t)\mathrm{OPT}_Z$, we obtain
% \[
% \frac{\SOL^I}{\mathrm{OPT}} \geq
% \frac{\alpha_{\mathrm{GW}}\mathrm{OPT}_Z + \frac{H}{2}}{\mathrm{OPT}_Z+H}
% =
% \frac{\alpha_{\mathrm{GW}}+\frac{t}{2}}{1+t},
% \qquad
% \frac{\SOL^F}{\mathrm{OPT}} \geq 
% \frac{\frac{\mathrm{OPT}_Z}{2}+H}{\mathrm{OPT}_Z+H}
% =
% \frac{\frac{1}{2}+t}{1+t}.
% \]
% Hence the returned value satisfies
% \[
% \frac{\max\{\SOL^I,\SOL^F\}}{\mathrm{OPT}} \geq \min_{t\geq 0} 
% \max\left\{ \frac{\alpha_{\mathrm{GW}}+\frac{t}{2}}{1+t}, \frac{\frac{1}{2}+t}{1+t} \right\}.
% \]

% The worst case occurs at the intersection point where the two expressions are equal:
% \[
% \alpha_{\mathrm{GW}}+\frac{t}{2}=\frac{1}{2}+t \quad\Longrightarrow\quad
% t = 2\Big(\alpha_{\mathrm{GW}}-\frac{1}{2}\Big)\approx 0.757134.
% \]
% Substituting into, for example, the second expression yields
% \[
% \frac{\max\{\SOL^I,\SOL^F\}}{\mathrm{OPT}} \geq \frac{\frac{1}{2}+t}{1+t} \approx \frac{0.5+0.757134}{1+0.757134} = 0.7154457.
% \]
\end{proof}

% ==========================================================
\section{A SOC-SDP relaxation upper-bounding the true quantum optimum}
% ==========================================================

\subsection{The anticommutation tradeoff inequality}
The key local geometric constraint comes from anticommutation. Intuitively, a single qbit cannot simultaneously have large $X$-polarization and large $Z$-polarization. In our setting, this translates into a quantitative tradeoff between $\langle X_i\rangle$ and any incident correlation $\langle Z_iZ_j\rangle$.

\begin{lemma}\label{anticomm}
Let $A,B$ be Hermitian operators with $A^2=B^2=\Id$ and $AB=-BA$. Then for every state $\rho$,
\[
\langle A\rangle_\rho^2 + \langle B\rangle_\rho^2 \le 1.
\]
\end{lemma}

\begin{proof}
Let $a=\langle A\rangle_\rho=\Tr[\rho A]$ and $b=\langle B\rangle_\rho=\Tr[\rho B]$, and define $M=aA+bB$. Using $A^2=B^2=\Id$ and $AB=-BA$, we get
\[
M^2 = a^2A^2 + b^2B^2 + ab(AB+BA) = (a^2+b^2)\Id,
\]
hence $\|M\|=\sqrt{a^2+b^2}$. On the other hand, because of the linearity of the trace
\[
\Tr[\rho M] = a\Tr[\rho A]+b\Tr[\rho B]=a^2+b^2.
\]
Since $|\Tr[\rho M]|\leq \|M\|_{\mathrm{op}} = \max\{|\lambda_{\min}(M)|, |\lambda_{\max}(M)|\}$ for Hermitian $M$, we obtain $a^2+b^2\leq \sqrt{a^2+b^2}$ i.e., $a^2+b^2 \leq 1$.
\end{proof}

\subsection{The relaxation}
We now define a convex relaxation that upper-bounds the \emph{true quantum optimum} $\OPT=\lambda_{\max}(H_{\mathrm{TFIM'}})$. The relaxation uses variables for (i) single-qbit $X$-expectations $x_i\approx \langle X_i\rangle$, and (ii) pairwise $ZZ$ correlations $c_{ij}\approx \langle Z_iZ_j\rangle$.
Lemma~\ref{anticomm} yields the SOC constraints relation between $x_i$ and $c_{ij}$ on each edge incident to $i$ by substituting $A = X_i$ and $B = Z_iZ_j$ and observing that they anticommute.

Recall from Section 2 that the edge constraints are
\[
\Pi_{ij}^{Z}=\frac{1}{2}\bigl(\Id - J_{ij}Z_iZ_j\bigr),
\qquad
\Pi_i^{X}=\frac{1}{2}\bigl(\Id + X_i\bigr).
\]
Accordingly, for each edge $\{i,j\}$ we define for convenience the signed variable
\[
t_{ij}:=-J_{ij}c_{ij},
\]
so that the (relaxed) edge constraint value is exactly $(1+t_{ij})/2$, matching $\Tr[\Pi_{ij}^{Z}\rho]$.

\begin{tcolorbox}
\begin{definition}[SOC-SDP relaxation]\label{def:socsdp}
Let $n=|V|$. The relaxation has variables $x_i\in[-1,1]$ and a Gram matrix $C=(c_{ij})$ with $c_{ii}=1$.
Let $u_1,\dots,u_n$ be unit vectors such that $c_{ij}=\ip{u_i}{u_j}$, and define
\[
t_{ij}:=-J_{ij}c_{ij}\in[-1,1].
\]
Maximize
\begin{equation}\label{sdpobj}
\SDP =   \sum_{\{i,j\}\in E} w_{ij}\,\frac{1+t_{ij}}{2} + \sum_{i\in V} h_i\,\frac{1+x_i}{2} = \sum_{\{i,j\} \in E} \SDP_E[i,j] + \sum_{i \in V}\SDP_X[i]
\end{equation}
subject to:
\begin{align}
C = (C_{ij}) = (c_{ij}) & \succeq 0, \\
\ip{u_i}{u_i} & = 1 \qquad \forall i\in V, \\
x_i & \in [-1,1]\qquad \forall i\in V, \\
x_i^2 + c_{ij}^2 & \leq 1,\quad x_j^2 + c_{ij}^2 \leq 1 \qquad \forall \{i,j\}\in E. \label{soc}
\end{align}
\end{definition}
\end{tcolorbox}

%\begin{remark}[Weighted instances]
%As previously mentioned, the objective is linear in the nonnegative weights $w_{ij},h_i$, so all approximation ratios below apply to weighted instances.
%\end{remark}

%\subsection{The relaxation upper-bounds $\lambda_{\max}(H)$ for all sign patterns}
The next lemma states that every quantum state induces feasible moments for the SOC-SDP with same objective function value i.e.,  the SOC-SDP is an upper bound on the true  $\OPT$ even for arbitrary signed/frustrated patterns $J_{ij}\in\{\pm1\}$.

\begin{lemma}\label{sdpopt}
For every state $\rho$ on $n$ qbits, there exists a feasible solution to Definition~\ref{def:socsdp} with objective value exactly $\Tr[\htfimp\rho]$.
Consequently,
\[
\SDP \geq \OPT=\lambda_{\max}(H_{\mathrm{TFIM'}}).
\]
This holds for \emph{arbitrary} sign patterns $J_{ij}\in\{+1,-1\}$.
\end{lemma}

\begin{proof}
Let us fix an arbitrary quantum state $\rho$ and let us define
\[
x_i := \Tr[\rho X_i],\qquad c_{ij} := \Tr[\rho Z_iZ_j].
\]
Since Pauli operators have operator norm $1$, we have $|x_i|\leq 1$ and $|c_{ij}|\leq 1$.

\noindent\textbf{PSD constraints:}
Let $C=(c_{ij})$ with $c_{ii}=1$. For any $b\in\mathbb{R}^n$,
\[
b^\top C b
=\sum_{i,j} b_i b_j\,\Tr[\rho Z_iZ_j]
=\Tr\left[\rho\left(\sum_i b_i Z_i\right)^2\right]\geq 0,
\]
so $C\succeq 0$ and admits a Gram representation $c_{ij}=\ip{u_i}{u_j}$ with $\|u_i\|=1$.

\noindent\textbf{SOC constraints:}
Fix an edge $\{i,j\}\in E$. Let $A=X_i$ and $B=Z_iZ_j$. Then $A^2=B^2=\Id$ and $AB=-BA$.
By Lemma~\ref{anticomm},
\[
\Tr[\rho X_i]^2 + \Tr[\rho Z_iZ_j]^2 \leq 1 \quad\Longleftrightarrow\quad
x_i^2 + c_{ij}^2\leq 1.
\]
The same argument with $X_j$ gives $x_j^2+c_{ij}^2\leq 1$.

\noindent\textbf{Objective function value:}
Using $\Pi_{ij}^{Z}=\frac12(\Id-J_{ij}Z_iZ_j)$ and $\Pi_i^{X}=\frac12(\Id+X_i)$ and linearity of trace,
\[
\Tr[\Pi_{ij}^{Z}\rho] = \frac{1 - J_{ij}\Tr[\rho Z_iZ_j]}{2}
= \frac{1 + (-J_{ij}c_{ij})}{2}=\frac{1+t_{ij}}{2},
\qquad
\Tr[\Pi_i^{X}\rho] = \frac{1+\Tr[\rho X_i]}{2}=\frac{1+x_i}{2}.
\]
Therefore the SOC-SDP objective equals $\Tr[\htfimp\rho]$. Taking $\rho$ to be an optimal eigenstate (ground state) for $\OPT$ yields $\SDP\geq \OPT$.
\end{proof}

In the analysis below, we will solve the SOC-SDP relaxation up to some arbitrary precision $\epsilon \geq 0$ to obtain a feasible solution consisting of vectors $x_i$ and a Gram matrix $(C)_{ij}$. We denote by SDP the objective function value obtained (over the instance we are solving) and we further decompose SDP into the edge Ising terms and the field terms:
\[
\SDP = \SDP_E + \SDP_X =\sum_{\{i,j\} \in E} \SDP_E[i,j] + \sum_{i \in V}\SDP_X[i]
\]

% =================================================================
\section{A 0.786 approximation algorithm based on SOC-SDP}
% ==========================================================

Our algorithm computes two randomized product-state candidates from the \emph{same} SOC-SDP solution and returns the better one. 
The point of having two candidates is that different instances can be edge-dominated or field-dominated:
\begin{itemize}
\item \textbf{Candidate A (pure-$Z$ solution):} ignores the SDP $x_i$ values and uses a Goemans-Williamson style hyperplane rounding on the edge constraints $Z_iZ_j$. This gives us a strong approximation guarantee on edge constraints but only a trivial factor $\nicefrac{1}{2}$ on the field constraints.
\item \textbf{Candidate B ($X$, $ZZ$ mixed solution):} this solution matches the SDP field marginals exactly (hence achieves factor $1$ on field constraints) and uses the remaining budget according to Lemma \ref{anticomm} in order to round the edge constraints.
\end{itemize}
We then take the better of the two values to get a $0.786$ approximation factor.

%%%%%%%%%%%%%%%%%%%%%%%%%%%%%%%%%%%%%%%%%%%%%%%%%%%
\subsection{A pure $Z$ solution}
%%%%%%%%%%%%%%%%%%%%%%%%%%%%%%%%%%%%%%%%%%%%%%%%%%%
Algorithm \ref{alg:A}, given some TFIM input instance for the transformed hamiltonian $H_{\mathrm{TFIM'}}$, solves the corresponding SOC-SDP to obtain a feasible solution, then ignores the $x_i$ variables and uses pure a $ZZ$-basis rounding.
\begin{algorithm}
\caption{Pure $ZZ$ solution}
\label{alg:A}
\textbf{Input:} Unit vectors $u_i$ from the SOC-SDP solution.
\begin{enumerate}
\item Sample $g \sim N(0,\Id)$. 
\item Set $s_i = \sgn(\ip{g}{u_i})$ for all $i$.
\item Output the product state $\rho=\bigotimes_i \rho_i$ where
\[
\rho_i := \frac{1}{2}\Big(\Id + s_i Z_i\Big),\]
(so we have that $\langle X_i\rangle=0, \langle Z_i\rangle = s_i$).
\end{enumerate}
\end{algorithm}
%\begin{definition}[Random hyperplane signs]
%Given some unit vectors $u_1,\dots,u_n$, we sample $g\sim N(0,I)$ and we set $s_i =\sgn(\ip{g}{u_i})\in\{\pm1\}$.
%\end{definition}
In this algorithm, we crucially use Grothendieck's identity (see for example Lemma 3.6.5, page 82 from \cite{Vershynin_2018}), which is described in the following lemma. 
\begin{lemma}[Grothendieck's Identity]\label{grothendieck}
Let $u_i,u_j$ be arbitrary unit vectors, $g\sim N(0,\Id)$ be a random vector, and $s_i,s_j$ be their random signs obtained by the hyperplane rounding according to $s_i = \sgn(\ip{g}{u_i})$. Then
\[
\E[s_i s_j] = K(\ip{u_i}{u_j}),
\text{ where }
K(t) =\frac{2}{\pi}\arcsin(t).
\]
\end{lemma}

%\begin{proof}
%Let $\theta=\arccos(\ip{u_i}{u_j})$ be the angle between $u_i$ and $u_j$.
%A random hyperplane through the origin separates $u_i$ and $u_j$ with probability $\theta/\pi$.
%Thus $\Pr[s_is_j=1]=1-\theta/\pi$ and $\Pr[s_is_j=-1]=\theta/\pi$, so $\E[s_is_j]=1-2\theta/\pi$.
%Using $\arcsin(t)=\frac{\pi}{2}-\arccos(t)$ gives $1-2\arccos(t)/\pi = \frac{2}{\pi}\arcsin(t)=K(t)$.
%\end{proof}

The following is also a standard inequality used in the analysis of the Goemans-Williamson {\scshape MaxCut} approximation algorithm: 
\begin{definition}\label{GW}
For $t\in[0,1]$ let
\[
R_A(t) =\frac{1+K(t)}{1+t},
\qquad
\alpha_{\mathrm{GW}} = \min_{t\in[0,1]} R_A(t) \approx 0.878567.
\]
\end{definition}

Now we analyze the value of each term for the solution produced by Algorithm \ref{alg:A} compared to its value in the SDP solution.  We start with the field terms in Lemma \ref{alg_A_fields} followed by the Ising terms in Lemma \ref{alg_A_edges}.

\begin{lemma}\label{alg_A_fields}
For each vertex $i$ we have that $\Tr[\Pi^X_i\rho]=\frac{1}{2}$. Hence we have
\[
\sum_{i\in V} h_i\,\Tr[\Pi^X_i\rho] =\frac{1}{2}\sum_{i\in V}h_i \geq  \frac{1}{2}\sum_{i\in V} h_i\,\frac{1+x_i}{2} = \frac{\SDP_X}{2},
\]
i.e. the total field contribution is at least $\frac{1}{2}\cdot \SDP_X$.
\end{lemma}

\begin{proof}
By construction $\Tr[\rho X_i]=0$, so $\Tr[\Pi^X_i\rho]=(1+0)/2=1/2\geq (1+x_i)/4$.
\end{proof}

\begin{lemma}\label{alg_A_edges}
Let $t_{ij}:=-J_{ij}\ip{u_i}{u_j}$. Then for each edge $\{i,j\}\in E$,
\[
 \E\big[w_{ij}\Tr[\Pi_{ij}^{Z}\rho]\big]
=
w_{ij} \frac{1+K(t_{ij})}{2} \geq w_{ij}\cdot \alpha_{\mathrm{GW}}\cdot \frac{1+t_{ij}}{2} = \alpha_{\mathrm{GW}} \cdot \SDP_E[i,j].
\]  Hence the total Ising contribution is at least $\alpha_{\mathrm{GW}} \cdot \SDP_E$.
\end{lemma}

\begin{proof}
Since $\rho$ is a product state with $\langle Z_i\rangle=s_i$, we have $\langle Z_iZ_j\rangle=s_is_j$ and
\[
\Tr[\Pi_{ij}^{Z}\rho]
=
\frac{1 - J_{ij}\langle Z_iZ_j\rangle}{2}
=
\frac{1 - J_{ij}s_is_j}{2}.
\]
Taking expectation and applying Grothendieck's Identity (Lemma \ref{grothendieck}),
\[
\E[-J_{ij}s_is_j] = -J_{ij}K(\ip{u_i}{u_j}) = K(-J_{ij}\ip{u_i}{u_j})=K(t_{ij}),
\]
since $K$ is odd. This gives $\E[\Tr[\Pi_{ij}^{Z}\rho]]=(1+K(t_{ij}))/2$.
The inequality with $\alpha_{\mathrm{GW}}$ holds by definition of $\alpha_{\mathrm{GW}}$ on $t\in[0,1]$. For $t\leq 0$ the ratio is at least $1$.
\end{proof}

Combining the two previous lemmas gives the following guarantee for Algorithm \ref{alg:A}:

\begin{theorem}\label{thm:A}
The product state $\rho^{(A)}$ output by Algorithm \ref{alg:A} satisfies
\[
\E[\Tr(\htfimp\rho^{(A)})] \geq \alpha_{\mathrm{GW}}\cdot \SDP_E \;+\; \frac{1}{2}\cdot \SDP_X,
\]
where $\SDP_E$ and $\SDP_X$ denote the edge and field parts of the SOC-SDP objective.
\end{theorem}

%%%%%%%%%%%%%%%%%%%%%%%%%%%%%%%%%%%%%%%%%%%%%%%%%%%
\subsection{A mixed $X$, $ZZ$ algorithm}
%%%%%%%%%%%%%%%%%%%%%%%%%%%%%%%%%%%%%%%%%%%%%%%%%%%

Algorithm \ref{alg:B} matches the SDP field marginals exactly and uses the remaining ``budget" according to Lemma \ref{anticomm} to round the $Z$-marginals for edges. Thus, it will be easy to see that it achieves approximation factor 1 on the field terms, and in the following (Lemma \ref{alg_B_edges}) we will bound the approximation guarantee on the Ising terms.

\begin{algorithm}[!ht]
\caption{$X$, $ZZ$ mixed solution}
\label{alg:B}
\textbf{Input:} SOC-SDP solution $(x_i,u_i)$ from Definition~\ref{def:socsdp}
\begin{enumerate}
\item For each vertex $i$, set $\alpha_i = \sqrt{1-x_i^2}$.
\item Sample $g \sim N(0,\Id)$.
\item Set $s_i = \sgn(\ip{g}{u_i})$.
\item Output the product state $\rho=\bigotimes_{i\in V}\rho_i$ where
\[
\rho_i = \frac{1}{2}\Big(\Id + x_i X_i + \alpha_i s_i Z_i\Big),\]
(so we have that $\langle X_i\rangle = x_i, \langle Z_i\rangle = s_i\alpha_i$).
\end{enumerate}
\end{algorithm}

Note that each $\rho_i$ is a valid pure quantum state since its Bloch vector has length $x_i^2+\alpha_i^2=1$.

%\begin{remark}
%Each $\rho_i$ is a valid pure qubit state because its Bloch vector has length $x_i^2+\alpha_i^2=1$.
%\end{remark}

\begin{lemma}\label{alg_B_fields}
For each vertex $i$, Algorithm~\ref{alg:B} satisfies
\[
\Tr[h_i \Pi^X_i\rho]=\frac{1+x_i}{2} = \SDP_X[i].
\]
Hence the total field contribution equals the SOC-SDP field contribution $\SDP_X$.
\end{lemma}

\begin{proof}
$\Pi_i^X=(\Id+X_i)/2$ acts only on qbit $i$ and $\rho$ is a product state, so
\[
\Tr[h_i \Pi^X_i\rho]=h_i \Tr[\Pi^X_i\rho_i]=h_i \frac{1+\Tr[\rho_i X_i]}{2}=h_i \frac{1+x_i}{2} = \SDP_X[i].
\]
\end{proof}

\begin{lemma}\label{Bedges}
Let $t_{ij}:=-J_{ij}\ip{u_i}{u_j}$ and $\alpha_i = \sqrt{1-x_i^2}$ as in Algorithm \ref{alg:B}. Then for every edge $\{i,j\}$,
\[
\E\big[\Tr[\Pi_{ij}^{Z}\rho]\big] = \frac{1+\alpha_i\alpha_j K(t_{ij})}{2}.
\]
\end{lemma}

\begin{proof}
Because $\rho = \bigotimes_{k} \rho_k$ is a product state, we have that $\langle Z_iZ_j\rangle_{\rho}= \Tr[\rho_i Z_i] \Tr[\rho_j Z_j] = (\alpha_i s_i)(\alpha_j s_j)$.
Therefore 
\[
\Tr[\Pi_{ij}^{Z}\rho]=\frac{1 - J_{ij}\alpha_i\alpha_j s_is_j}{2}.
\]
Taking expectation and applying Grothendick's Identity (Lemma~\ref{grothendieck}) as in Lemma~\ref{alg_A_edges} gives $\E[-J_{ij}s_is_j]=K(t_{ij})$ and hence the desired bound.
\end{proof}

We also need the following Lemma that relates the SDP Ising values $t_{ij}$ with the``leftover" rounded values $\alpha_i \alpha_j$:

\begin{lemma}\label{alpha_t}
If $t_{ij}\ge 0$ then $\alpha_i\alpha_j \ge t_{ij}^2$.
\end{lemma}

\begin{proof}
The SOC anticommutation constraint \eqref{soc} gives us $x_i^2+c_{ij}^2\le 1$ from which we see that $\alpha_i=\sqrt{1-x_i^2}\geq |c_{ij}|$, and similarly $\alpha_j\ge |c_{ij}|$.
Thus $\alpha_i\alpha_j\ge c_{ij}^2=t_{ij}^2$.
\end{proof}

\begin{definition}\label{beta}
For $t\in[0,1]$ let
\[
R_B(t) =\frac{1+t^2K(t)}{1+t},
\qquad
\beta  = \min_{t\in[0,1]} R_B(t).
\]
Basic calculus (setting the derivative to zero) shows that $\beta \approx 0.716775$ achieved at $t\approx 0.58$.
\end{definition}
\noindent

Finally, we can analyse the contribution of the Ising terms as compared to the SDP solution.

\begin{lemma}\label{alg_B_edges}
For each edge $\{i,j\}$, Algorithm \ref{alg:B} satisfies
\[
\E\big[\Tr[\Pi_{ij}^{Z}\rho]\big] \ge \beta \cdot \frac{1+t_{ij}}{2} = \beta \cdot \SDP_E[i,j].
\] Hence the total Ising contribution is at least $\beta \cdot \SDP_E$.
\end{lemma}

\begin{proof}
For $t_{ij}\in[0,1]$, Lemmas ~\ref{Bedges} and~\ref{alpha_t} give
\[
\E[\Tr[\Pi_{ij}^{Z}\rho]]
=\frac{1+\alpha_i\alpha_j K(t_{ij})}{2}
\ge \frac{1+t_{ij}^2K(t_{ij})}{2}
= R_B(t_{ij})\cdot \frac{1+t_{ij}}{2}
\ge \beta\cdot \frac{1+t_{ij}}{2}.
\]
For $t_{ij}\leq 0$, note that $K(t_{ij})\le 0$ and $\alpha_i\alpha_j\le 1$, hence
$1+\alpha_i\alpha_j K(t_{ij}) \geq 1+K(t_{ij})$.
Moreover, for $t\in[-1,0]$ we have $K(t)\ge t$, so $1+K(t_{ij})\ge 1+t_{ij}$, i.e. the ratio is at least $1$.
\end{proof}

Combining Lemmas \ref{alg_B_fields} and \ref{alg_B_edges} gives the following approximation guarantee for Algorithm \ref{alg:B}:

\begin{theorem}\label{thm:B}
The product state $\rho^{(B)}$ output by Algorithm \ref{alg:B} satisfies
\[
\E[\Tr(\htfimp\rho^{(B)})] \geq \beta\cdot \SDP_E + 1\cdot \SDP_X.
\]
\end{theorem}

\subsection{Combining Algorithms \ref{alg:A} \& \ref{alg:B}}

\noindent\textbf{Proof idea.}
Algorithm \ref{alg:A} is strong on edges (factor $\alpha_{\mathrm{GW}}$) but weak on fields (trivial factor of only $\nicefrac{1}{2}$) while Algorithm \ref{alg:B} is perfect on fields (factor $1$) but weaker on edges (factor $\beta$). The SOC-SDP value decomposes into an ``edge portion'' $\SDP_Z$ and a ``field portion'' $\SDP_X$. Let $p$ to be the fraction of the SDP coming from edges allows us to express each candidate's guarantee into an affine function of $p$. Taking the better candidate yields the maximum of these affine functions and the worst-case over instances is the minimum of that maximum over $p\in[0,1]$.

\begin{theorem}[Two-candidate product rounding]\label{thm:A+B}
The algorithm which on an input TFIM instance executes Algortihms A and B and returns whichever of the two candidate outputs has higher energy achieves approximation ratio
\[\gamma = \frac{1}{2} + \frac{1}{2}\cdot \frac{\alpha_{\mathrm{GW}}-1/2}{\alpha_{\mathrm{GW}}+1/2 - \beta}\approx 0.786\]

(where $\alpha_{\mathrm{GW}}$ and $\beta$ are as defined in Definitions~\ref{GW} and \ref{beta} respectively).

% Let $\alpha_{\mathrm{GW}}$ be as in Definition \ref{GW} and $\beta \approx 0.716775$ as in Definition~\ref{beta}. Then there is a randomized polynomial-time algorithm that outputs the best of Algorithm \ref{alg:A} and Algorithm \ref{alg:B}, and achieves approximation ratio
% \[
% \gamma = \min_{p\in[0,1]} \max\left\{\frac{1}{2}+\Big(\alpha_{\mathrm{GW}}-\frac{1}{2}\Big)p,~~ 1-\Big(1-\beta \Big)p\right\}
% \]
% with respect to $\OPT=\lambda_{\max}(H_{\mathrm{CSP}})$, i.e.,
% \[
% \E[\Tr(H_{\mathrm{CSP}}\rho)] \ge \gamma\cdot \OPT.
% \]
%Numerically, $\gamma\approx 0.786016$.
\end{theorem}

\begin{proof}
Recall that we have \[\SDP=\SDP_E+\SDP_X.\]
Let \[p=\frac{\SDP_E}{\SDP}\in [0,1],\] the fraction of the optimal solution energy represented by the Ising part, and correspondingly $1-p=\frac{\SDP_X}{\SDP}$ is the fraction represented by the field part.

% We expand the argument in two steps.

% Firstly we split the SOC-SDP objective into ``edge mass'' and ``field mass''. As we have mentioned in the previous parts, we write the SOC-SDP objective (Definition \ref{def:socsdp}) as the sum of the Ising (edge) and fields parts:
% \[
% \SDP = \SDP_E + \SDP_X,
% \]
% where
% \[
% \SDP_E = \sum_{\{i,j\}\in E} w_{ij}\frac{1+t_{ij}}{2},
% \qquad \mbox{ and } \qquad
% \SDP_X = \sum_{i\in V} h_i\frac{1+x_i}{2}.
% \]
% Both $\SDP_E$ and $\SDP_X$ are nonnegative because $w_{ij},h_i\ge 0$ and $(1+t_{ij})/2,(1+x_i)/2\in[0,1]$. We als define the edge fraction
% \[
% p = \frac{\SDP_E}{\SDP_E+\SDP_X}=\frac{\SDP_E}{\SDP}\in[0,1] \mbox{ and  similarly } 1-p = \frac{\SDP_X}{\SDP} \in [0,1].
% \]
% Intuitively $p$ measures how much of the SDP value comes from satisfying edge constraints versus field constraints.

Now we can express the two candidate guarantees in terms of $p$:

\smallskip
\noindent\emph{Candidate A:}
From Theorem \ref{thm:A} we have
\[
\E[\Tr(H_{\mathrm{CSP}}\rho^{(A)})] \geq 
\alpha_{\mathrm{GW}}\cdot \SDP_E + \frac{1}{2} \cdot \SDP_X \Rightarrow 
\frac{\E[\Tr(H_{\mathrm{CSP}}\rho^{(A)})]}{\SDP} \geq 
\frac{1}{2} + \Big(\alpha_{\mathrm{GW}}-\frac{1}{2}\Big) \cdot p.
\]

\smallskip
\noindent\emph{Candidate B:}
From Theorem \ref{thm:B} we have
\[
\E[\Tr(H_{\mathrm{CSP}}\rho^{(B)})] \geq \beta\cdot \SDP_E + 1\cdot \SDP_X \Rightarrow 
\frac{\E[\Tr(H_{\mathrm{CSP}}\rho^{(B)})]}{\SDP} \geq 1-(1-\beta) \cdot p.
\]

\medskip
Finally, we can take  the best of the two solutions and calculate the worst-case guarantee. Then we have that
%\[
%\E[\Tr(H_{\mathrm{CSP}}\rho)] \geq  \max\left\{
%\E[\Tr(H_{\mathrm{CSP}}\rho^{(A)})], \E[\Tr(H_{\mathrm{CSP}}\rho^{(B)})] \right\}.
%\]
%Dividing by $\SDP$ and applying the two lower bounds derived above gives us
\[
\frac{\E[\Tr(H_{\mathrm{CSP}}\rho)]}{\SDP} \geq  
\max\left\{
\frac{1}{2} + \left(\alpha_{\mathrm{GW}}-\frac{1}{2}\right) \cdot p, 1-(1-\beta)\cdot p \right\}.
\]

Note that the two terms inside the maximum are respectively increasing and decreasing functions of $p$.  They intersect at $p^\star=\frac{1/2} {\alpha_{\mathrm{GW}}+1/2-\beta}$, giving an overall worst-case approximation factor of $\gamma=1-(1-\beta)p^*\approx 0.786016$.

%\[
%\frac{\E[\Tr(H_{\mathrm{CSP}}\rho)]}{\SDP} \geq \min_{p\in[0,1]} \max\left\{
%\frac{1}{2} + \left(\alpha_{\mathrm{GW}}-\frac{1}{2}\right) \cdot p, 
%1-(1-\beta) \cdot p\right\} = \gamma.
%\]

%Putting everything together gives 
%\[
%\E[\Tr(H_{\mathrm{CSP}}\rho)] \geq  \gamma\cdot \SDP \geq \gamma\cdot \OPT,
%\]
%\noindent

\end{proof}

% ==========================================================
\section{An improved $0.8156$ factor based on interpolation}
% ==========================================================

For the algorithm in the previous section we considered the two options of either setting each qubit to the purely $Z$ direction in the Bloch sphere with signs chosen by rounding the SDP, or on the other hand taking the $x_i$ values from the SDP and then using the remaining slack for the $Z$ direction as before. In this section we interpolate between these two cases, by first scaling the $x_i$ values by a constant parameter $q\in [0,1]$ and then using the remaining slack for the $Z$ direction.  

\begin{algorithm}[h!]
\caption{Interpolated $X$, $ZZ$ combined solution}
\label{alg:C}
\textbf{Input:} SOC-SDP solution $(x_i,u_i)$ from Definition~\ref{def:socsdp}, parameter $q \in [0,1]$.
\begin{enumerate}
\item For each vertex $i$, set $\alpha_i(q) = \sqrt{1-(qx_i)^2}$.
\item Sample $g \sim N(0,I)$.
\item Set $s_i = \sgn(\ip{g}{u_i})$.
\item Output the product state $\rho(q)=\bigotimes_{i\in V}\rho_i$ where
\[
\rho_i = \frac{1}{2}\Big(\Id + qx_i X_i + \alpha_i(q) s_i Z_i\Big),\] (so we have that $\langle X_i\rangle = qx_i, \langle Z_i\rangle = s_i\alpha_i(q), s_i \in \{-1, +1\}$).
\end{enumerate}
\end{algorithm}

This generalises Algorithms A ($q=0$) and B ($q=1$).  We find that the optimal algorithm is obtained with $q\approx 0.6312$, yielding an approximation factor of $\gamma \approx 0.8156$.

We proceed with the analysis, in the same way as in Algorithm \ref{alg:B}. We start by analyzing the field term contribution to the objective function. 

\begin{lemma}
For each field term (vertex) $i \in V$ we have that 
\[
\Tr[h_i\Pi_i^X \rho(q)] = h_i\frac{1+ \langle X_i \rangle_{\rho(q)}}{2} = h_i\frac{1+qx_i}{2} \geq \frac{1+q}{2} \cdot \SDP_X[i].
\]  Hence the total field contribution is at least $\frac{1+q}{2} \cdot \SDP_X$.
\end{lemma}

\begin{proof}
The equalities are immediate since by construction $\langle X_i \rangle_{\rho(q)} = qx_i$.  Now \[h_i\frac{1+qx_i}{2} = \frac{1+qx_i}{1+x_i} \cdot \SDP_X[i] \geq \frac{1+qx_i+q+x_i}{2(1+x_i)} \cdot \SDP_X[i] = \frac{1+q}{2} \cdot \SDP_X[i],\]
since $0 \leq q,x_i \leq 1 \Rightarrow q+x_i\leq 1 + qx_i$.

% Since, by construction, $\langle X_i \rangle_{\rho(q)} = qx_i$, we have $\Tr[\Pi_i^X \rho(q)] = \frac{1+ \langle X_i \rangle_{\rho(q)}}{2} = \frac{1+qx_i}{2}$.  We also want to find a $\zeta$ such that
% \[
% \frac{1+q x_i}{2} \geq \zeta \cdot \frac{1+x_i}{2} \Rightarrow \zeta \leq \frac{1+qx_i}{1+x_i} \leq \frac{1+q}{2},
% \]
% where the last inequality follows because $|x_i| \leq 1$ and $q < 1$. We substitute $\SDP_X = \nicefrac{(1+x_i}{2}$ and the claim follows. 
\end{proof}

Now we proceed, exactly as before, with the contribution of the Ising terms (edges) to the solution energy.
\begin{lemma}
For every Ising term (edge $\{i,j\} \in E$) we have that 
\[
\E[w_{ij} \Tr[\Pi_{ij}^Z \rho(q)]] \geq \beta(q) \cdot \SDP_E[i,j], \mbox{ where } \beta(q) = \min_{t \in [0,1]} \frac{1+ \Big( (1-q^2) + q^2t^2 \Big) K(t)}{1+t}
\] Hence the total Ising contribution is at least $\beta{q} \cdot \SDP_E$.
\end{lemma}

\begin{proof}
We proceed as in Lemma \ref{alg_B_edges}, but with the new values of $\alpha_i$.

Fix an edge $\{i,j\}$ and we have 
\[
\Tr[\Pi_{ij}^Z \rho(q)] = \frac{1-J_{ij} \langle Z_iZ_j \rangle }{2} = \frac{1 + \alpha_i(q)\alpha_j(q) (-J_{ij} \cdot s_i s_j)}{2}.
\]
As before, taking expectation over the hyperplane rounding and using again Grothendieck's Identity, we have that
\[
\E [-J_{ij} \cdot s_is_j] = K(t_{ij}) \mbox{ where } t_{ij} = -J_{ij} \langle u_i, u_j \rangle)
\]
and so
\[
\E[\Tr[\Pi_{ij}^Z \rho(q)]] = \frac{1+\alpha_i(q) \alpha_j(q) K(t_{ij})}{2}.
\]
Using Lemma \ref{anticomm}, for each endpoint incident in an edge $\{i,j\}$ we have that $x_i^2 \leq 1-t_{ij}^2$ and also $x_j^2 \leq 1- t_{ij}^2$. We use these to derive bounds on $\alpha_i(q), \alpha_j(q)$:
\[
\alpha_i(q) = \sqrt{1 - q^2x_i^2 } \geq \sqrt{1 - q^2 (1-t_{ij}^2)} = \sqrt{(1-q^2) + q^2t_{ij}^2},
\]
and similarly for $\alpha_j(q)$ giving us the bound
\[
\alpha_i(q) \alpha_j(q) \geq (1-q)^2 +q^2t_{ij}^2 
\]
and hence 
\begin{equation}\label{eq:edgebound}
\E[\Tr[\Pi_{ij}^Z \rho(q)]] \geq \frac{1+ \Big((1-q)^2 +q^2t_{ij}^2 \Big)K(t_{ij})}{2}.
\end{equation}
Since $\SDP_E[i,j]=w_{ij}\frac{1+t_{ij}}{2}$, dividing (\ref{eq:edgebound}) by $\frac{1+t_{ij}}{2}$ and minimising with respect to $t=t_{ij}$ gives the result.
\end{proof}

Note that at $q=0$ we get 
\[
\frac{1+ \Big( (1-q^2) + q^2t^2 \Big) K(t)}{1+t} = \frac{1+K(t)}{1+t} \mbox{ and so } \beta(0) = \alpha_{\mathrm{GW}},
\]
i.e. Algorithm \ref{alg:A}, and at $q=1$
\[
\frac{1+ \Big( (1-q^2) + q^2t^2 \Big) K(t)}{1+t} = \frac{1+t^2K(t)}{1+t} \mbox{ and so } \beta(1) = \beta = 0.716775,
\]
i.e. Algorithm \ref{alg:B}.

Now we proceed with the balancing analysis. As before let $p = \frac{\SDP_E}{\SDP} \in [0,1]$ and $1-p = \frac{SDP_X}{\SDP} \in [0,1]$, with $\SDP = \SDP_E + \SDP_X$. By the previous two Lemmas the approximation factor guaranteed by our algorithm depends on $p,q$ as follows:
\[
Q(p,q) = \beta(q) \cdot \SDP_E + \frac{1+q}{2} \cdot \SDP_X = p \cdot \beta(q) + (1-p)\cdot \frac{1+q}{2} = \frac{1+q}{2} + \Big[\beta(q) - \frac{1+q}{2} \Big]\cdot p. 
\]

We will choose $q$ such that this is constant as a function of $p$, giving the following:

\begin{theorem}\label{interp-opt}
    Algorithm \ref{alg:C} with $q$ taken as the root of the equation $\beta(q)=\frac{1+q}{2}$ ($q^*\approx 0.6312$) achieves approximation factor $\beta(q^*) \approx 0.8156$.
\end{theorem}

In principle we could choose a scaling factor $q$ depending on the value of $p$ obtained by solving the SDP, and this would give worst-case approximation factor of  
\[
\delta = \min_{p \in [0,1]} \max_{q \in [0,1]} Q(p,q). 
\]

However, it turns out that the optimal value of $q$ varies continuously as a function of $p$ between 0 and 1.  Hence in particular there is some $p$ ($p\approx 0.7094$) for which the optimal value of $q$ is $q^*$ as in Theorem \ref{interp-opt}, and so in fact $\delta = \beta(q^*)$ and there is no worst-case advantage in choosing $q$ depending on $p$.

The optimal $q$ and the resulting approximation factor depending on the value of $p$ are shown in Figure \ref{fig:optqrat}; note the phase transition between $q=1$ and $q<1$, which occurs at $p\approx 0.602$.

\begin{figure}[h]
\centering
\includegraphics[width=0.47\textwidth]{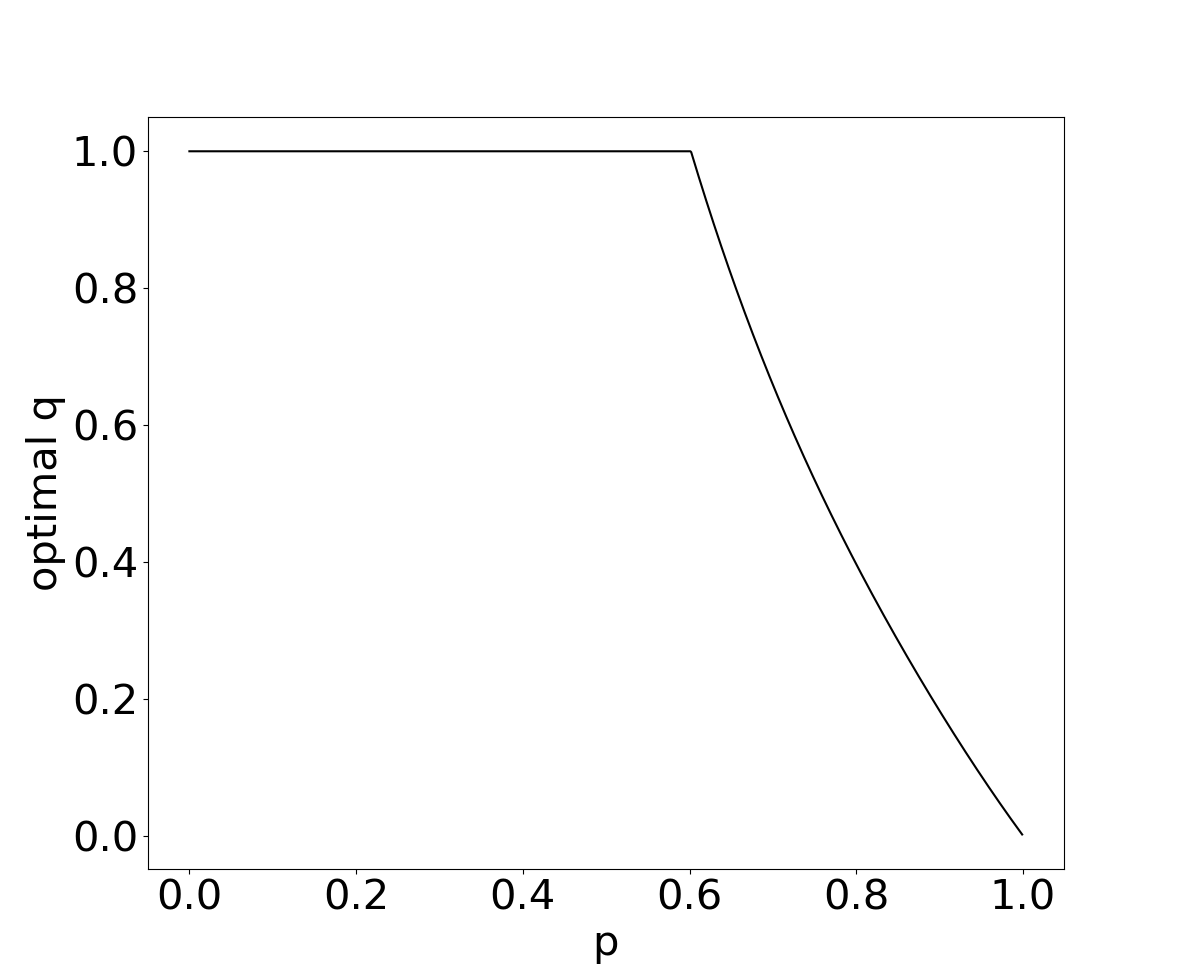} \includegraphics[width=0.52\textwidth]{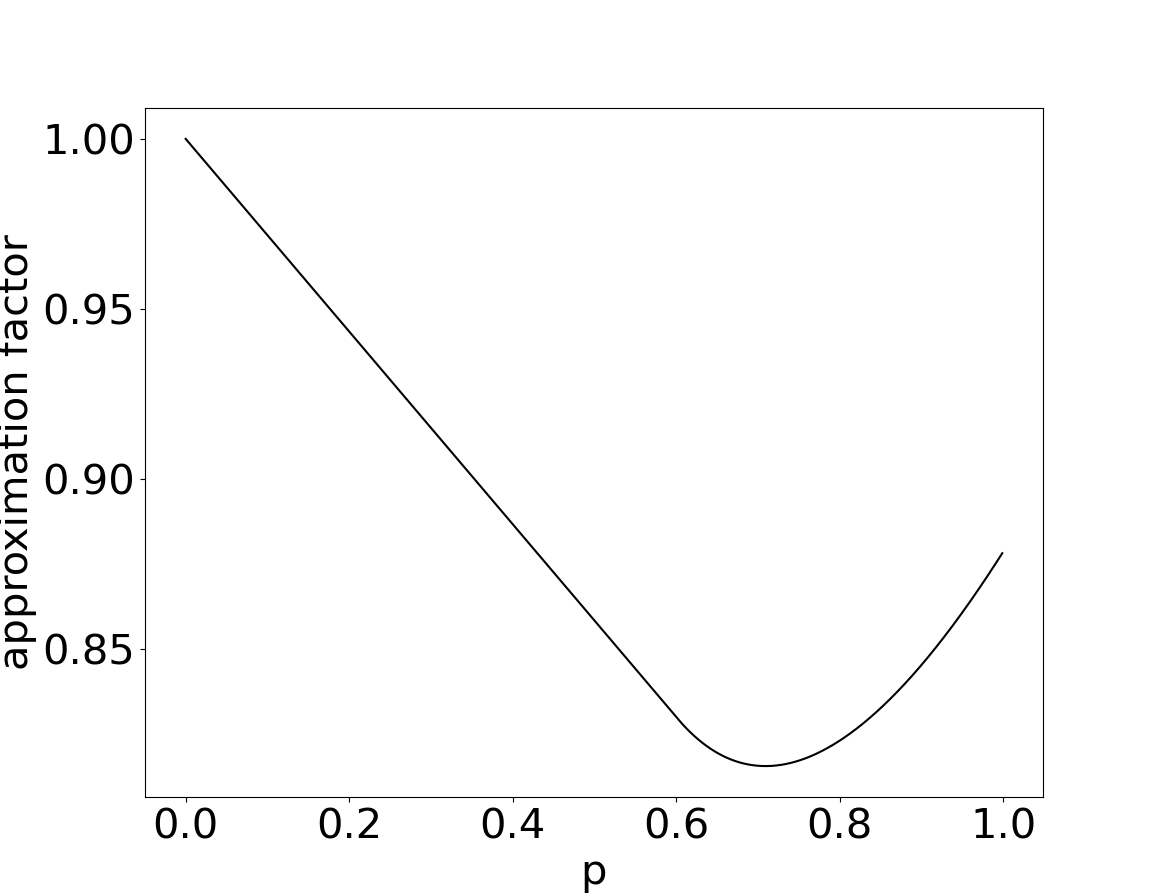}
\caption{Optimal $q$ and resulting approximation factor for varying $p$.}\label{fig:optqrat}
\end{figure}

\section{A product-state $169/180$ optimality gap on a triangle instance}
% ==========================================================

In this section, we give a concrete $3$-qbit purely ferromagnetic TFIM instance where product states provably cannot achieve the true optimum. This gives an upper bound for any method that is restricted to output product states.

Let $V=\{0,1,2\}$ and $E=\{\{0,1\},\{1,2\},\{0,2\}\}$.
Set unit edge weights $w_{ij}=1$ and uniform field weights $h_i=g$ with
\[
g =\frac{3}{5}.
\]
We take the signs
\[
J_{01}=J_{12}=J_{02}= +1.
\]

Let $H_t$ denote the corresponding $H_{\mathrm{CSP}}$ Hamiltonian on this instance. This is a $3$ qbit Hamiltonian and as such a $2^3\times 2^3$ Hermitian matrix. A direct diagonalization
gives the spectrum
\[
\left\{\frac{18}{5},\ \frac{16}{5}\ (\times 2),\ \frac{13}{5}\ (\times 2),\ \frac{8\pm \sqrt{19}}{5},\ \frac{4}{5}\right\},
\]
so the maximum eigenvalue is $\frac{18}{5}$.

The interesting direction is to upper bound the value achievable by any product state.

\begin{lemma}
Let $\OPT_{\mathrm{prod}} =\max_{\rho\ \mathrm{product}}\Tr[H_t\rho]$. For this instance 
\[
\OPT_{\mathrm{prod}}=\frac{169}{50}.
\]
\end{lemma}

\begin{proof}
We have:
\[
\OPT_{\mathrm{prod}} = \max_{\rho = \rho_1 \otimes \rho_2 \otimes \rho_3} \Tr[H_t \rho].
\]

\noindent 
% For $H_t$ every product state is fully determined by the local Bloch vector representations of each qbit. We will write the product state energy of our Hamiltonian as a function of these Bloch sphere parameters and find the optimum. 

% We will first argue that whatever maximum must be attained at a pure state. If we fix all qbits except, say, qbit $i$ for some $i$ i.e., $\rho_{-i} = \otimes_{k \neq i}\rho_k$, then the term $\Tr[H_t (\rho_i \otimes \rho_{-i})] = \Tr[A_i \rho_i]$ is linear in $\rho_i$. The set of feasible points $\{ \rho_i \mbox{ is PSD and } \Tr[\rho_i] = 1  \}$ is a compact convex set with extreme points the pure states $\ket{\chi} \bra{\chi}$. This is because a mixed state is a convex combination of pure states and thus if the optimum is attined at a mixed state then there must exist a pure state which is at least as good as this mixed state.  Since any linear function over a convex set assumes its optimum at an extreme point, we can conclude that the for any fixed $\rho_{-i}$ the best choice for $\rho_i$ must be a pure state. 

We will represent each $\rho_i$ by its Bloch vector $(x_i,y_i,z_i)$; note that trivially the optimum is reached at a pure state so we have $x_i^2+y_i^2+z_i^2=1$.  Further, we may assume that $x_i\geq 0$ and $y_i=0$, since otherwise replacing $(x_i,y_i,z_i)$ with $(\sqrt{x_i^2+y_i^2},0,z_i)$ improves the solution.  Hence also we have $x_i = \sqrt{1-z_i^2}$.

% Now we argue further reduce the parameters. Any qubit $i$ corresponds to a unit vector $(x_i, y_i, z_i)$ ($x_i^2 +y_i^2 +z_i^2 = 1$) in the Bloch Sphere. Since our Hamiltonian only uses the terms $x_i = \langle X_i \rangle$ and $z_i = \langle Z_i \rangle$ through the products $z_iz_j$, and since $h_i \geq 0$ for all $i$, we can take $y_i = 0$ and $x_i \geq 0$. This gives that $x_i = \sqrt{1 - z_i^2}$.

Now we will write the objective as a function of $(z_0,z_1,z_2)$.
Since $\langle Z_iZ_j \rangle = z_iz_j$ and $\langle X_i \rangle =x_i$ we have
\[
\Tr[H_t\rho] = \sum_{\{i,j\}\in E}\frac{1-J_{ij}z_i z_j}{2} +\frac{g}{2}\sum_{i=0}^2 (1+x_i).
\]
Putting $J_{ij}=+1$ for all edges, $g=\nicefrac{3}{5}$ and $x_i = \sqrt{1-z_i^2}$ we obtain
\[
F(z)
=
\frac{1}{2}(1-z_0z_1)+\frac{1}{2}(1-z_1z_2)+\frac{1}{2}(1-z_0z_2) +\frac{3}{10}\sum_{i=0}^2 (1+\sqrt{1-z_i^2}),
\]
and so our optimization problem is transformed to a simple optimization problem in the cube $[-1,1]^3$:
\[
\OPT_{\mathrm{prod}} = \max_{|z_i| \leq 1} F(z_0,z_1,z_2).
\]

Since $g>0$, an optimal solution cannot have $|z_i|=1$ (equivalently $x_i=0$) since otherwise pertubing that qbit slightly so that $z_i$ decreases by $\epsilon$ and $x_i$ increases by $\sqrt{\epsilon}$ changes the edge terms only by $\approx \epsilon^2$ but increases the field contribution by $\epsilon$.

Now we can differentiate $F$ with respect to each $z_i$. The first-order stationary conditions gives us the system 
\[
x_0(z_1+z_2)=-gz_0,\qquad
x_1(z_0+z_2)=-gz_1,\qquad
x_2(z_0+z_1)=-gz_2,
\qquad (g=3/5).
\]
Solving this system over $|z_i|<1$ gives the following critical points:
\begin{itemize}
\item[(i)] $z_0=z_1=z_2=0$ (hence $x_0=x_1=x_2=1$);
\item[(ii)] up to relabeling, one coordinate is $0$ and the other two are opposite:
\[
z_a=0,\qquad z_b=-z_c= +/- \frac{4}{5},
\qquad
x_a=1,\qquad x_b=x_c=\frac{3}{5}.
\]
\end{itemize}
%(Note to self: these are easy to check directly as follows: if $z_a=0$ then the first equation forces $z_b = -z_c$ and the remaining equations force $x_b = x_c = g$, i.e. |z_b|=|z_c|=\sqrt{1-g^2}=4/5. The all-zero point is immediate.)

Now we evaluate $F$ on these candidates solutions. At $(z_0,z_1,z_2)=(0,0,0)$ we have $x_i=1$ and thus
\[
F=\frac{3}{2}+\frac{3}{10}\cdot 3\cdot 2=\frac{33}{10}=3.3.
\]
At (ii), say $(z_0,z_1,z_2)=(0,4/5,-4/5)$ with $(x_0,x_1,x_2)=(1,3/5,3/5)$, we get
\[
\text{edges}
=\frac12(1-0)+\frac12\Bigl(1-\Bigl(-\frac{16}{25}\Bigr)\Bigr)+\frac12(1-0)
=\frac{91}{50},
\]
and
\[
\text{fields}
=\frac{3}{10}\Bigl[(1+1)+\Bigl(1+\frac35\Bigr)+\Bigl(1+\frac35\Bigr)\Bigr]
=\frac{78}{50}.
\]
So $F=91/50+78/50=169/50$. 

Also, there are also solutions with all $z_i \neq 0$, but numerical optimization shows that they give worse bounds with respect to F.
\end{proof}
A product state that matches the value $169/50$ is:
\[
\ket{\rho_{\mathrm{OPT}_{\mathrm{prod}}}}
=
\frac{1}{\sqrt{2}}
\left( \ket{0} + \ket{1} \right)
\otimes
\left( \sqrt{0.9}\,\ket{0} + \sqrt{0.1}\,\ket{1} \right)
\otimes
\left( \sqrt{0.1}\,\ket{0} + \sqrt{0.9}\,\ket{1} \right),
\]
\begin{corollary}
For this instance,
\[
\frac{\OPT_{\mathrm{prod}}}{\OPT}
=\frac{169/50}{18/5}
=\frac{169}{180}
\approx 0.938888.
\]
Consequently, no algorithm restricted to output product states can guarantee a worst-case approximation ratio
exceeding $169/180$ for TFIM instances.
\end{corollary}

\bibliographystyle{plainurl}
\bibliography{references}

\end{document}